\documentclass[aps,pra]{revtex4-1}
\usepackage{graphicx} % Include figure files
\usepackage{color}
\usepackage{dcolumn}  % Align table columns on decimal point
\usepackage{amssymb,latexsym}
\usepackage{amsbsy} % for \boldsymbol
\usepackage{amsthm}
\usepackage{stmaryrd}
\usepackage{verbatim}
\usepackage{amsmath}
\usepackage{fullpage}
\usepackage{enumitem}

\newcommand{\calE}{\mathcal{E}}
\newcommand{\calC}{\mathcal{C}}
\newcommand{\calV}{\mathcal{V}}

\newcommand{\Tr}{\operatorname{Tr}}
\newcommand{\dom}{\operatorname{dom}}

\newcommand \bra[1] {\langle {#1} |}
\newcommand \ket[1] {| {#1} \rangle}
\newcommand \braket[1] {\langle {#1} \rangle}
\newcommand{\RR}{\mathbb{R}}
\newcommand{\NN}{\mathbb{N}}
\newcommand{\RRinf}{\mathbb{R}\cup\{+\infty\}}
\newcommand{\RRminf}{\mathbb{R}\cup\{-\infty\}}
\newcommand{\jp}{{\vec{j}_\text{p}}}
\renewcommand{\vec}[1]{\mathbf{#1}}
\DeclareMathOperator*{\argmin}{arg\,min}

\newtheorem{theorem}{Theorem}
\newtheorem{proposition}{Proposition}

\newtheorem{corollary}{Corollary}
\newtheorem{definition}{Definition}

\begin{document}

\title{Ground-state densities from the Rayleigh--Ritz variation principle and from density-functional theory}

\author{Simen Kvaal}
\email{simen.kvaal@kjemi.uio.no}
\affiliation{Centre for Theoretical and Computational Chemistry, Department of Chemistry, University of Oslo, P.O. Box 1033 Blindern, N-0315 Oslo, Norway}

\author{Trygve Helgaker}
\email{t.u.helgaker@kjemi.uio.no}
\affiliation{Centre for Theoretical and Computational Chemistry, Department of Chemistry, University of Oslo, P.O. Box 1033 Blindern, N-0315 Oslo, Norway}

\date{\today}

%Use showkeys class option if keyword

\begin{abstract}
  The relationship between the densities of ground-state wave functions (i.e., the minimizers of
  the Rayleigh--Ritz variation principle) and the
  ground-state densities in density-functional theory (i.e., the minimizers of the Hohenberg--Kohn variation principle)
is studied within the framework of convex
  conjugation, in a generic setting covering  molecular systems,
  solid-state systems, and more. Having introduced admissible density
functionals as functionals that produce the exact ground ground-state energy for a
  given external potential by minimizing over densities in
  the Hohenberg--Kohn variation principle, necessary sufficient conditions on such 
  functionals are established to ensure that the Rayleigh--Ritz
  ground-state densities and the Hohenberg--Kohn ground-state densities
  are identical.
  We apply the results to molecular systems in the Born--Oppenheimer
  approximation. For any given potential $v
  \in L^{3/2}(\RR^3) + L^{\infty}(\RR^3)$, we establish a one-to-one
  correspondence between the mixed ground-state densities of the
  Rayleigh--Ritz variation principle and the mixed ground-state
  densities of the Hohenberg--Kohn variation principle when the Lieb
  density-matrix constrained-search universal density functional is
  taken as the admissible functional. A
  similar one-to-one correspondence is established between the pure
  ground-state densities of the Rayleigh--Ritz variation principle and
  the pure ground-state densities obtained using the Hohenberg--Kohn
  variation principle with the Levy--Lieb pure-state
  constrained-search functional. In other words, all physical
  ground-state densities (pure or mixed) are recovered with these
  functionals and no false densities (i.e., minimizing
  densities that are not physical) exist. The importance of topology (i.e.,
  choice of Banach space of densities and potentials) is emphasized
  and illustrated. The relevance of these results for
  current-density-functional theory is examined.
\end{abstract}

\maketitle

\section{Introduction}
\label{sec:introduction}

The concept of modeling an electron gas and more
generally all electronic systems using only the density $\rho$ goes
back to the early days of quantum mechanics and the independent work
of Tomas and Fermi in 1927 \cite{Thomas1927,Fermi1927}. However, it
was only with the seminal work of Hohenberg and Kohn in 1964 that density-functional theory (DFT) was
shown to be---in principle, at least---an exact
theory\;\cite{Hohenberg1964}.  In 1983, Lieb formulated DFT rigorously
for electronic systems in $\RR^3$, using concepts from convex
analysis\;\cite{Lieb1983}. Today, DFT is the most popular
computational method for many-electron systems in chemistry and
solid-state physics.

In an external electrostatic scalar potential $v$, the ground-state energy $E(v)$ of an $N$-electron system is 
obtained from the \emph{Rayleigh--Ritz variation principle} 
\begin{equation}
  E(v) := \inf_\psi \braket{\psi|\hat{H}(v)|\psi} \label{eq:RR-char},
\end{equation}
where $\hat{H}(v)$ is the Hamiltonian and the infimum
extends over all $L^2$-normalized wave functions for which the
expectation value makes sense. With each minimizing ground-state wave
function $\psi$ in Eq.~(\ref{eq:RR-char}), there is an 
associated ground-state density $\rho$.
In DFT, the ground-state energy $E(v)$ is instead obtained from 
the \emph{Hohenberg--Kohn variation principle}, 
\begin{equation}
  E(v) = \inf_{\rho} \bigl( F_0(\rho) + ( v \vert \rho) \bigr),
  \label{eq:dft-basic}
\end{equation}
where $F_0(\rho)$ is the \emph{universal functional} and the minimization
is over all densities for which the interaction
$(v \vert \rho) = \int \! v(\mathbf r) \rho(\mathbf r)  \mathrm d \mathbf r$ 
is meaningful. Note that we have not specified the spatial domain
nor the nature of the external potential $v$. Thus, the minimization problems~(\ref{eq:RR-char}--\ref{eq:dft-basic}) are so far
not fully defined from a mathematical perspective. Moreover,
the functional $F_0$ is not unique---any functional $F_0$ that gives
the correct ground-state energy in the Hohenberg--Kohn variation principle is said to be `admissible'.
Common expositions of DFT use the \emph{Levy--Lieb constrained-search}
functional $F_\text{LL}$ \cite{Levy1979,Lieb1983},  
the \emph{density-matrix constrained-search functional} $F_\text{DM}$
\cite{Lieb1983}, and the \emph{Lieb functional} $F$,
all giving the same ground-state energy $E(v)$ for a
large class of potentials. Specializing to electrons in $\RR^3$ and
assuming that the external potential $v$ is sufficiently well-behaved, these functionals are given by the following expressions:
\begin{align}
  F_\text{LL}(\rho) &:= \inf \!\left\{ \braket{\psi|\hat{H}_0 |\psi} \,
    \big|\, \psi \in H^1_\text{S}(\RR^{3N}), \, \Vert \psi \Vert_2 = 1,\, \psi\mapsto\rho  \right\}, \label{eq:FLL-def} \\
  F_\text{DM}(\rho) &:= \inf \!\left\{ \Tr(\hat{H}_0 \Gamma) \, \big| \,
  \Gamma = \textstyle \sum\nolimits_k \lambda_k \ket{\psi_k}\bra{\psi_k}, \, \textstyle \sum\nolimits_k \lambda_k = 1, \,
  \lambda_k \geq 0, \, 
    \psi_k \in H^1_\text{S}(\RR^{3N}),\, \Vert \psi_k \Vert_2 = 1,\,
    \psi\mapsto\rho\, \!\right\}\!, \label{eq:FDM-def} \\
  F(\rho) &:= \!\sup_{v\in X'} ( E(v) - (v | \rho) ). \label{eq:lieb-vp}
\end{align}
Here $H^1_\text S(\RR^{3N})$ is the first-order Sobolev space $H^1_\text S(\RR^{3N})$ with proper permutation
symmetry due to spin,
$\hat{H}_0 = \hat{T} + \hat{W}$ is the sum of the
kinetic and inter-electron interaction operators, and 
$\psi \mapsto \rho$ and $\Gamma \mapsto \rho$ indicate that the wave function $\psi$ and density matrix $\Gamma$, respectively, have density $\rho$. 
In the definition of the Lieb functional, $X$ is a Banach space in
which densities are embedded, with $X'$ as its dual, consisting of
potentials. We note that the
Lieb functional is by construction lower semi-continuous and convex,
being the conjugate function to $E$ in the sense of convex
analysis.

Whereas the density functionals $F_\text{LL}$, $F_\text{DM}$ and $F$ are all admissible and therefore give the same total energy 
in the Hohenberg--Kohn variation principle, they may have different minimizing densities. 
In this article, we study the relationship between the densities in the Hohenberg--Kohn and Rayleigh--Ritz variation principles.
In particular, we establish what conditions must be imposed on an admissible density functional to
ensure that the ground-state densities (i.e., the minimizing densities) in the Hohenberg--Kohn
variation principle are the same as the densities of the ground-state wave functions
(i.e., the minimizing wave functions) in the Rayleigh--Ritz variation principle. 

While our treatment covers the standard DFT setting outlined above, it is 
motivated by a study of current-density-functional theory (CDFT),
where the $N$-electron system is subject to an external magnetic vector
potential $\vec{A}$ in addition to the scalar potential $v$
\cite{Vignale1987,Vignale1988,Laestadius2014}, introducing the
paramagnetic current density $\jp \in \vec{L}^1(\RR^3)$ as an
additional variable in the Hohenberg--Kohn variation principle:
\begin{equation}
  E(v,\vec{A}) = \inf_{\rho,\jp} \left( F_0\left(\rho,\jp\right) + \left( v +
  \tfrac{1}{2}A^2 \mid \rho\right) + \left( \vec{A} \mid \jp \right) \right),\label{eq:cdft-hk}
\end{equation}
where $\left( \vec{A} \mid \jp \right) = \int\! \vec A(\mathbf r)
\cdot \jp(\mathbf r) \,\mathrm d \mathbf r$.  
%In CDFT, no Hohenberg--Kohn-type theorem can aid us in finding a `Hohenberg--Kohn
%functional', see the discussion in Ref.\;\cite{Tellgren2012}. Instead, the point of departure is the constrained-search formulation. 
The CDFT generalization of the Levy--Lieb functional is
the \emph{Vignale--Rasolt functional} $F_\text{VR}(\rho,\jp)$; there
is likewise a generalization of the density-matrix constrained-search
functional $F_\text{DM}$ and the Lieb functional $F$ to include a current dependence. The
mathematical properties of these functionals are not as well
understood as in the standard DFT setting. The results presented in
this article are in part meant to be a step towards such an understanding.

The article is structured as follows. In Sec.\,\ref{sec:abstract}, we
describe DFT from a generic and abstract point of view, introducing
arbitrary admissible density functionals $F_0$ and some concepts of convex
analysis. We obtain several results that characterize the relationship
between minimizers of the Rayleigh--Ritz
variation principle and of the Hohenberg--Kohn variation principle. In
Sec.\,\ref{sec:dft-application} we apply our findings to the standard
DFT of atoms and molecules in the Born--Oppenheimer approximation. The
importance of the topology of the underlying density space is emphasized. In
Sec.\,\ref{sec:cdft-application} we discuss CDFT, and finally, in
Sec.\,\ref{sec:conclusion} we summarize and draw some conclusions.

\section{DFT from an abstract point of view}
\label{sec:abstract}
\subsection{Admissible universal density functionals}

Except for the trivial case of one electron, an explicit formula for
any of the standard density functionals used in the Hohenberg--Kohn
variation principle in Eq.~(\ref{eq:dft-basic}) is not
known. Moreover, the mathematical analysis of Eq.~\eqref{eq:dft-basic}
is difficult without further assumptions. In his work
\cite{Lieb1983}, Lieb placed DFT for electronic systems in $\RR^3$ on
a firm mathematical ground using the language of \emph{convex
  analysis} \cite{VanTiel,EkelandAndTemam}, the natural setting for
problems such as Eq.~\eqref{eq:dft-basic}.  The starting point is to
embed the densities in a Banach space $X$ and to consider potentials
in the dual space $X'$, thereby making the interaction
$(v\,\vert\,\rho)$ well-defined and continuous in both arguments. We then
obtain the ground-state energy as a map $E : X' \to \RRminf$ given by
\begin{equation}
E(v) = \inf_{\rho\in X} \left( F_0(\rho) + (v|\rho) \right), \quad \forall v\in X',
\label{eq:admissible-dft}
\end{equation}
where $F_0 : X \to \RRinf$ is the universal density functional. Being the
pointwise infimum of a collection of affine maps of the form $v \mapsto 
(v \,\vert\, \rho) + F_0(\rho)$, the ground-state energy in Eq.~(\ref{eq:admissible-dft}) is automatically
\emph{concave and upper semi-continuous} with respect to the topology
on $X'$. Note that we allow $E$ and $F_0$ to be infinite. The
\emph{effective domain} of a function is the set where the function is
finite---for example, $\dom(E) = \{ v \in X' \mid E(v)>-\infty\}$.

In Ref.~\cite{Lieb1983}, Lieb (and Simon) proved several important
results in the context where $X = L^1(\RR^3) \cap L^3(\RR^3)$, with
dual $X' = L^{3/2}(\RR^3) + L^\infty(\RR^3)$.  However, this choice of
$X$ is not unique, and the derived results may depend on $X$. We
emphasize that, even if our notation suggests the standard DFT
setting, it covers also CDFT and other settings.

A function $F_0: X \to \RRinf$ is said to be an \emph{admissible} density functional if, for each potential $v \in X^\prime$,
it gives a correct ground-state energy by the Hohenberg--Kohn variation principle 
in~Eq.~(\ref{eq:admissible-dft})---that is, if it produces identical
results with the Rayleigh--Ritz variation
principle in Eq.~\eqref{eq:RR-char}.
Note that we do not require that there exists a minimizing 
density $\rho$ for a given potential $v$, even in those cases where
$v$ supports a ground state in the Rayleigh--Ritz variation principle
in Eq.\;(\ref{eq:RR-char}). It is sufficient that a minimizing
sequence $\{\rho_n\}$ can be found.  Clearly, for any
pair $(\rho,v) \in X \times X^\prime$, an admissible density
functional and the ground-state energy satisfy the \emph{Fenchel inequality}
\begin{equation}
E(v) \leq  F_0(\rho) + (v \vert \rho) .
\label{eq:Fenchel}
\end{equation}
We are here interested in characterizing those pairs $(\rho,v)$ that saturate Fenchel's inequality,
$E(v) =  F_0(\rho) + (v \vert \rho)$; in other words, those $\rho\in X$ that are minimizers in
Eq.~\eqref{eq:admissible-dft} for a given $v\in X'$.

Since the universal functional is in general not differentiable (see
Refs.~\cite{Lammert2005,Lieb1983} for the standard DFT case), we
cannot write down an Euler equation for the solution of the
minimization problem Eq.~\eqref{eq:admissible-dft}. Moreover, even if
$F_0$ were differentiable, the Euler equation would in general be a necessary but
not sufficient condition for a global minimum in
Eq.~\eqref{eq:admissible-dft}. Instead, we use the concept of
\emph{subdifferentiation} to characterize the solution.  The
admissible density functional $F_0$ is said to be \emph{subdifferentiable} at
$\rho\in X$ if $F_0(\rho) \in \mathbb R$ and if there exists  $u \in X'$, known
as a \emph{subgradient} of $F_0$ at $\rho$, such that
\begin{equation}
  F_0(\rho') \geq F_0(\rho) + ( u  \mid \rho'-\rho), \quad \forall \rho'\in
  X, \label{eq:subdiff-def}
\end{equation}
meaning that $F_0$ touches the affine map 
$\rho' \mapsto F_0(\rho') + ( u \,\vert\, \rho' - \rho)$ at $\rho$ and lies nowhere below it.
The subgradient is thus a generalization of the concept of the
slope of a tangent to the graph of $F_0$ at $\rho$. 
The \emph{subdifferential} $\partial^- F_0(\rho)$ is the (possibly empty) set of all subgradients of $F_0$ at $\rho$:
\begin{equation}
\partial^- F_0(\rho) = \left\{ \,u \in X^\prime \, \left\vert  \, F_0(\rho') \geq F_0(\rho) + ( u  \vert \rho'-\rho), 
\; \forall \rho'\in X,  \; F_0(\rho) \in \mathbb R\, \right.\right\} .
\label{eq:subdiff}
\end{equation}
This set is convex \cite{VanTiel}.
Note that the subdifferential is the empty set whenever $F_0(\rho) = + \infty$. Assuming that $-v $ is a subgradient of $F_0$ at $\rho$, we obtain
by simple rearrangements
\begin{align}
-v \in \partial^- F_0(\rho) & \iff F_0(\rho^\prime) \geq F_0(\rho) + ( -v \,\vert\, \rho^\prime - \rho ) \nonumber \\
                    & \iff F_0(\rho) + ( v \,\vert\, \rho )  \le F_0(\rho^\prime)  + ( v \,\vert\, \rho^\prime) \nonumber \\
                    & \iff F_0(\rho) + ( v \,\vert \rho ) = \inf_{\rho^\prime \in X} \left(F_0(\rho^\prime) + (v \,\vert \rho^\prime) \right)  = E(v),
\label{eq:toFencheq}
\end{align}
We have thus shown the following sufficient and necessary condition
for a minimizing density in the Hohenberg--Kohn variation principle in Eq.\;\eqref{eq:admissible-dft}:
\begin{proposition}\label{thm:equivalence}
  Let $F_0 : X \to \RRinf$ be an admissible density functional so that
  Eq.~\eqref{eq:admissible-dft} holds. Let $v\in X'$ and $\rho \in X$
  be given. Then, 
  \begin{align}
    E(v) = F_0(\rho) + ( v \vert \rho ) \iff -v \in \partial^- F_0(\rho)  .
    \label{eq:toSubEq}
  \end{align}
\end{proposition}

\subsection{Lieb's universal density functional}
\label{sec:lieb}

From the Fenchel inequality in Eq.~(\ref{eq:Fenchel}), we obtain by a
trivial rearrangement the equivalent inequality
\begin{equation}
F_0(\rho) \geq  E(v) - (v \,\vert\, \rho),
\label{eq:FenchelL}
\end{equation}
stating that each admissible density functional $F_0(\rho)$ is an upper
bound to $\rho \mapsto E(v) - (v \,\vert\, \rho)$ with respect to all variations in
$v \in X^\prime$.  We now define the \emph{Lieb universal density
  functional} $F(\rho)$ as the \emph{least upper bound} to $E(v) - (v
\,\vert\, \rho)$ for all $v \in X^\prime$:
\begin{equation}
F(\rho) = \sup_{v \in X^\prime}  \left( E(v) - (v \,\vert\, \rho) \right), \quad \forall \rho \in X,
\label{eq:FL}
\end{equation}
which in 
the following will be referred as the \emph{Lieb variation principle}.
The Lieb functional is clearly a lower bound to all admissible density functionals:
\begin{equation}
F(\rho) \leq F_0(\rho).
\label{eq:Flower}
\end{equation}
Since the Lieb functional by construction also satisfies the Fenchel inequality in Eq.~(\ref{eq:Fenchel}), we obtain:
\begin{proposition}
The Lieb functional is an admissible density functional:
\begin{equation}
E(v) = \inf_{\rho \in X} \left( F(\rho) + (v \,\vert\, \rho) \right), \quad \forall v \in X'.
\label{eq:EL}
\end{equation}
\end{proposition}
\begin{proof}
  $E(v) \leq \inf_{\rho\in X} \left( F(\rho) + (v\,\vert\,\rho)\right)
  \leq \inf_{\rho\in X} \left( F_0(\rho) + (v\,\vert\,\rho)\right) = E(v)$.
\end{proof}

The Lieb functional is related to the ground-state energy in a special,
symmetrical manner---compare the Lieb variation principle in Eq.\,(\ref{eq:FL}) with 
the Hohenberg--Kohn variation principle Eq.\,(\ref{eq:EL}).  In
the language of convex analysis, $E$ and $F$ are said to be
\emph{skew conjugate functions} \cite{Rockafellar1968}: the function $E$ is concave and upper
semi-continuous, whereas its skew conjugate function $F$ is
convex and lower semi-continuous.  Convexity of $F$ and
concavity of $E$ mean that, for each pair $\rho_1,\rho_2 \in X$, each pair $v_1,v_2 \in X^\prime$, and each $\lambda \in (0,1)$,
we have
\begin{align}
F(\lambda \rho_1 + (1-\lambda) \rho_2) &\leq \lambda F(\rho_1) + (1-\lambda) F(\rho_2), \\
E(\lambda v_1 + (1-\lambda) v_2) &\geq \lambda E(v_1) + (1-\lambda) E(v_2), 
\end{align}
whereas lower semi-continuity of $F$ and upper semi-continuity of $E$ imply that
\begin{align}
\liminf_{\rho \to \rho_0} F(\rho) &\geq F(\rho_0), \\
\limsup_{v \to v_0} E(v) &\leq E(v_0).
\end{align}
These properties follow straightforwardly from Eqs.\;(\ref{eq:FL})
and~(\ref{eq:EL}).  It is a fundamental result of convex analysis that \emph{there
is a one-to-one correspondence between all lower semi-continuous
convex functions on $X$ and all upper semi-continuous concave
functions on $X^\prime$}
\cite{VanTiel,EkelandAndTemam,Rockafellar1968}, the correspondence
being as in Eqs.~\eqref{eq:FL} and \eqref{eq:EL}.  It follows that the Lieb functional $F$ is not only a lower bound to all admissible
density functionals but also
the \emph{{only} admissible density functional that is lower
semi-continuous and convex with respect to the topology on $X$}. 
It is a trivial but important observation that each property and each feature of $E$ are, in some manner, exactly reflected in the properties
and features of $F$ and vice versa. 
Note, however, that the Lieb functional depends explicitly on $X$, just like $E$ depends on $X^\prime$.

Just like there may happen to be no minimizing density in the
Hohenberg--Kohn variation principle, there may happen to be no maximizing
potential in the Lieb variation principle. To characterize maximizing
potentials, we introduce superdifferentiability by analogy with
subdifferentiability.
The ground-state energy $E$ is said to be
\emph{superdifferentiable} at $v\in X^\prime$ if there exists an element $\rho \in X$, 
known as a \emph{supergradient} of $E$ at $v$, such that
\begin{equation}
  E(v^\prime) \leq E(v) + ( v^\prime - v  \mid \rho), \quad \forall v' \in X^\prime. \label{eq:subdiff-defE}
\end{equation}
The \emph{superdifferential} $\partial^+ E(v)$ is the (possibly empty) convex set of all supergradients of $E$ at $v$:
\begin{equation}
\partial^+ E(v) = \left\{ \,\rho \in X \, \left\vert  \, E(v') \leq E(v) + ( \rho  \vert v'-v), 
\; \forall v'\in X^\prime,  \; E(v) \in \mathbb R\, \right.\right\}.
\label{eq:subdiff}
\end{equation}
In exactly the same manner that we proved Eq.~(\ref{eq:toSubEq}), we obtain the following necessary and sufficient conditions 
for the existence of a maximizing potential in the Lieb variation principle:
\begin{align}
E(v) = F(\rho) + ( v \vert \rho ) \iff \rho \in \partial^+ E(v) .
\label{eq:toSupEqF}
\end{align}
Note carefully that this result holds only for the Lieb functional, not for
an arbitrary admissible density functional.
Combining Eqs.~(\ref{eq:toSubEq}) and~(\ref{eq:toSupEqF}), we arrive at 
the following characterization of the optimality condition in
Eqs.~\eqref{eq:FL} and \eqref{eq:EL}:
\begin{proposition}\label{thm:subdiffs-conj}
If $F : X\to\RRinf$ is the Lieb functional and $E : X' \to \RRminf$ is the ground-state energy, then
  \begin{align}
    E(v) = F(\rho) + ( v \,\vert\, \rho ) \iff -v \in \partial^- F(\rho)
    \iff \rho \in \partial^+ E(v) . 
    \label{eq:toEF}
  \end{align}
\end{proposition}

For a general admissible density functional $F_0$, these conditions
are not equivalent: there may exist $\rho \in \partial^+ E(v)$ such
that $-v \notin \partial^- F_0(\rho)$ when $F_0 \neq F$.
On the other hand, the converse statement, $-v
\in \partial^- F_0(\rho) \implies \rho \in \partial^+ E(v)$, holds for
any admissible density functional $F_0$. To prove this, assume that $-v
\in \partial^- F_0(\rho)$. According to Eq.~(\ref{eq:toSubEq}), we then
have $F_0(\rho) = E(v) - ( v \,\vert\, \rho )$.  At the same time, the
Fenchel inequality holds for any admissible density functional
$F_0$:
\begin{equation}
E(u) \leq F_0(\rho) + ( u \,\vert\, \rho ), \quad \forall u \in X^\prime .
\label{eq:fyf}
\end{equation}
Substituting $F_0(\rho) = E(v) - ( v \,\vert\, \rho )$ in this inequality, we obtain
\begin{equation}
E(u) \leq E(v) + ( u -v \,\vert\, \rho ), \quad \forall u \in X^\prime ,
\end{equation}
implying that $\rho \in \partial^+ E(v)$. 
We have thus proved the following:
\begin{proposition}\label{thm:subdiffs0}
If $F_0 : X\to\RRinf$ is an admissible functional and $E : X' \to \RRminf$ the ground-state energy, then
  \begin{align}
    E(v) = F_0(\rho) + ( v \,\vert\, \rho ) \iff -v \in \partial^- F_0(\rho) \implies \rho \in \partial^+ E(v) .
    \label{eq:toEFad}
  \end{align}
\end{proposition}

In general, we have $F \leq F_0$. However, according to the following proposition,
an admissible density functional $F_0(\rho)$ can differ from $F(\rho)$ only 
when $\partial^- F_0(\rho) = \emptyset$: 
\begin{proposition}\label{thm:subdiffs01}
If $F_0 : X\to\RRinf$ is an admissible functional and $F: X\to\RRinf$ the Lieb functional, then 
\begin{equation}
-v \in \partial^- F_0(\rho) \iff F_0(\rho) = F(\rho) \;\wedge\;  -v \in \partial^- F(\rho)  .
\end{equation}
\end{proposition}
\begin{proof}
Let us assume that $-v \in \partial^- F_0(\rho)$. From Eq.~(\ref{eq:toEFad}),
we then have $E(v) = F_0(\rho) + ( v \,\vert\, \rho )$. Substituting this result into Eq.~(\ref{eq:toEF}), we then find 
that $F(\rho) = F_0(\rho)$ and $-v \in \partial^- F(\rho)$. Conversely, if
$F(\rho) = F_0(\rho)$ and $-v \in \partial^- F(\rho)$ hold, 
then we have $E(v) = F_0(\rho) + ( v \,\vert\, \rho )$ by Eq.~(\ref{eq:toEF}), from which
$-v \in \partial^- F_0(\rho)$ follows by Eq.~(\ref{eq:toEFad}).
\end{proof}

\subsection{Constrained-search density functionals}
\label{sec:constrained-search}

So far, 
we have not established a connection between the minimizing densities in the Hohenberg--Kohn
variation principle in Eq.~(\ref{eq:dft-basic}) and the minimizing wave functions in the Rayleigh--Ritz variation principle in~Eq.\;(\ref{eq:RR-char}).
For instance, we have not shown 
that each $\rho \in \partial^+ E(v)$ is the ground-state density associated with some 
wave function $\psi$.  Introducing a generic (abstract) constrained-search density functional, we
consider in this section the relation between minimizers in the Hohenberg--Kohn and
Rayleigh--Ritz variation principles.

Let $S$ be a set with elements $s$ called \emph{states} and let
$d : S \to X$ be a map from $S$ to the Banach space of densities $X$. The
image 
\begin{equation}
\rho_s = d(s) \label{ds}
\end{equation}
is called the \emph{density} of $s$. 
However, no physical meaning is attached to $s\in S$ or to $\rho \in
X$ at this point; they are mathematical objects---for example, $\rho$
may be a (current) density or a reduced density matrix, while $s$ may
be a pure state or a density operator.  As before, $X'$ is the dual
space of $X$, meaning that $v\in X'$ if and only if $v : X \to \RR$ is
a linear and continuous operator. Thus, the pairing $(v\,\vert\,\rho)$
is separately linear and continuous in $v$ and $\rho$. We call $v$ a
\emph{potential} but again no physical meaning is attached at present.

Assume next that, for every $v\in X'$, we are given a map $\calE_v : S \to \RR$ of the form
\begin{equation}
  \calE_v(s) = \calE_0(s) + (v \mid \rho_s), \label{ecal}
\end{equation}
where $\calE_0 : S \to \RR$ is bounded below, meaning that there exists an $M\in \RR$ such that
\begin{equation}
\calE_0(s) \geq M,  \quad \forall s\in S.
\end{equation}
We call $\calE_v(s)$  the \emph{expectation value} of the
total energy when the system is in the state $s$ and influenced by the
potential $v$. 
This setting covers all standard formulations of DFT and CDFT, including both 
pure-state and density-matrix formulations---in particular, we do not
assume linearity of $\calE_0(s)$ and $d(s)$ in $s$.

Next, define the \emph{ground-state energy} $E : X' \to \RRminf$ by the Rayleigh--Ritz variation principle:
\begin{equation}
  E(v) := \inf_{s\in S} \calE_v(s) = \inf_{s\in S} \left( \calE_0(s) + (v\mid\rho_s) \right). \label{eq:E-def}
\end{equation}
Observing that $v$ couples to $\rho_s$ only during minimization, it
makes sense to define the \emph{constrained-search functional} $F_\textnormal{CS} : X \to
\RRinf$ as the map
\begin{equation}
  F_\textnormal{CS}(\rho) = \inf \left\{ \calE_0(s) \; | \; s \mapsto \rho \right\}, \label{eq:F0-def}
\end{equation}
which takes the value $+\infty$ whenever there is no $s\in S$ with
$\rho_s = \rho$. 
 Combining Eqs.\;(\ref{eq:E-def})
and~(\ref{eq:F0-def}), we obtain the Hohenberg--Kohn variation
principle
\begin{equation}
  E(v) = \inf_{\rho \in X} \left( F_\textnormal{CS}(\rho) + (v \mid \rho) \right). \label{eq:cs}
\end{equation}
This result proves:
\begin{proposition}
  The constrained-search functional $F_\textnormal{CS}$ in Eq.\,\eqref{eq:F0-def} is an
  admissible density functional for the energy $E$ in~Eq.\,\eqref{eq:E-def}.
\end{proposition}

\subsection{Characterization of ground states}
\label{sec:ground-states}

Suppose
that $v \in X^\prime$ is such that there exists $s\in S$ for which the infimum in Eq.~\eqref{eq:E-def} is a minimum:
\begin{equation}
  E(v) = \inf_{s'} \calE_v(s') = \calE_v(s),
\label{ecal2}
\end{equation}
meaning that there exists a \emph{ground state} $s \in \argmin_{s^\prime} \calE_v(s^\prime)$ for $v$. Let now $v'\in X'$ be arbitrary, and compute
\begin{equation}
  E(v') = \inf_{s'\in S} \calE_{v'}(s') \leq \calE_{v'}(s) 
= \calE_0(s) + (v'|\rho_s) 
= \calE_v(s) + (v' -v|\rho_s) 
= E(v) + (v'-v|\rho_s),
\end{equation}
where we have used Eqs.\;(\ref{ecal}) and~(\ref{ecal2}). We have thus proved the following proposition:
\begin{proposition}
If $s \in S$ is a ground state for $v$ in Eq.~\eqref{ecal}, then $\rho_s$ is a supergradient of  $E$ in Eq.~\eqref{eq:E-def} at $v$:
\begin{equation}
s \in \argmin_{s^\prime} \calE_v(s^\prime) \implies \rho_s \in \partial^+ E(v).
\end{equation}
\end{proposition}
The following proposition establishes an important relationship between ground states and subgradients of $F_\textnormal{CS}$:
\begin{proposition}
Let $F_\textnormal{CS} : X \to \RRinf$ be the constrained-search functional~\eqref{eq:F0-def}.
If $s$ is the ground state for some potential $v \in X^\prime$ with density $\rho_s$,
then $-v\in\partial^- F_\textnormal{CS}(\rho_s)$. 
\end{proposition}
\begin{proof}
  Let $s' \in S$ with $\rho_{s'} = \rho_s$. According to Eq.~\eqref{eq:E-def}, we then have
\begin{equation}
E(v) = \calE_0(s) + (v|\rho_s) \leq \calE_0(s') + (v|\rho_s). 
\end{equation}
Subtracting $(v|\rho_s)$ from both sides, taking the infimum on the right-hand side, and using the definition in~Eq.~(\ref{eq:F0-def}), we obtain
\begin{equation}
\calE_0(s) \leq \inf_{s'\mapsto \rho_s} \calE_0(s') = F_\textnormal{CS}(\rho_s) \leq \calE_0(s).
\end{equation}
Thus, if there is a ground state $s\in S$ for $v\in X'$, then
\begin{equation}
  F_\textnormal{CS}(\rho_s) + (v|\rho_s) = \calE_0(s) + (v \vert \rho_s) = \calE_v(s) = E(v) 
\leq F_\textnormal{CS}(\rho) + (v|\rho), \quad \forall \rho \in X.
\end{equation}
Rearranging, we obtain
\begin{equation}
  F_\textnormal{CS}(\rho) \geq F_\textnormal{CS}(\rho_s) - (v|\rho-\rho_s), \quad \forall \rho \in
  X, \label{eq:fsub}
\end{equation}
demonstrating that $-v \in \partial^- F_\textnormal{CS}(\rho_s)$ and completing the proof.
\end{proof}
Conversely, does the subgradient relation $-v\in\partial^- F_\textnormal{CS}(\rho)$ imply that there exists a ground state
$s\mapsto\rho$? To answer this question affirmatively, we must assume that $F_\textnormal{CS}$ is 
\emph{expectation valued}, defined as follows:
\begin{definition}[Expectation-valued constrained-search functional]
  A constrained-search functional $F_\textnormal{CS}$ is called \emph{expectation
    valued} if, for every $\rho$ with $F_\textnormal{CS}(\rho)<+\infty$, there exists an
  $s_\rho \in S$ such that
  \begin{equation}
    F_\textnormal{CS}(\rho) = \inf_{s\mapsto\rho} \calE_0(s) = \calE_0(s_\rho), \quad
    s_\rho \mapsto\rho. \label{eq:infmin}
  \end{equation}
\end{definition}
In other words, if $F_\textnormal{CS}$ is expectation valued, then
the infimum in Eq.~\eqref{eq:F0-def} is a minimum if
$F_\textnormal{CS}(\rho)<+\infty$, implying that $F_\textnormal{CS}(\rho)$ is the expectation
value $\calE_0(s)$ of some state $s\mapsto\rho$.
If $F_\textnormal{CS}$ is expectation valued, it follows immediately that 
\begin{equation}
-v \in \partial^- F_\textnormal{CS}(\rho) \implies E(v) = F_\textnormal{CS}(\rho) + (v|\rho) \implies
  E(v) = \calE_0(s_\rho) + (v|\rho) = \calE_v(s_\rho) \implies s_\rho \in 
\argmin_{s^\prime} \calE_v(s^\prime) 
\end{equation}
and therefore that there exists a ground state $s\mapsto\rho$ of $v$.
Summarizing, we have proved the following:
\begin{proposition}\label{thm:subdiffs1}
  Suppose that $F_\textnormal{CS} : X \to \RRinf$ is an expectation-valued constrained-search functional. 
Then, for each $v \in X'$, we have
  \begin{equation}
s_\rho \in \argmin_{s^\prime} \calE_v(s^\prime) \iff
-v\in \partial^- F_\textnormal{CS}(\rho) \implies \rho \in \partial^+ E(v).
  \end{equation}
\end{proposition}
Remark: Proposition~\ref{thm:subdiffs1} for an expectation-valued constrained-search functional should be compared with
Proposition~\ref{thm:subdiffs0} for a general admissible density functional. Whereas,
for a general admissible density functional $F_\textnormal{CS}$, the
condition $E(v) = F_\textnormal{CS}(\rho) + (v|\rho)$ does {not} imply the existence
of a ground state $s\mapsto \rho$, this implication does hold for 
for an \emph{expectation-valued constrained-search functional}.
Compare also with
Proposition~\ref{thm:subdiffs-conj}, where the admissible density
functional is assumed to be convex and lower semi-continuous (i.e., the Lieb functional).

\begin{proposition}\label{thm:subdiffs3}
  Suppose that an expectation-valued constrained-search functional $F_\textnormal{CS} : X \to \RRinf$ is convex and lower
  semi-continuous so that $F_\textnormal{CS} =F$. 
Then, for each $v \in X'$, we have
  \begin{equation}
s_\rho \in \argmin_{s^\prime} \calE_v(s^\prime) \iff
-v\in \partial^- F_\textnormal{CS}(\rho) \iff \rho \in \partial^+ E(v).
  \end{equation}
\end{proposition}
\begin{proof}
  Combine Propositions~\ref{thm:subdiffs-conj} and \ref{thm:subdiffs1}. 
\end{proof}

\subsection{Differentiability of $E(v)$}
\label{sec:diffability}

We consider here conditions for differentiability of $E : X' \to \RRminf$.
From Eq.~(\ref{eq:cs}), it follows that $E$
is a concave and upper semi-continuous function. Upper
semi-continuity is, however, a rather weak property. What is
needed to make $E$ \emph{continuous}? It is a fact that a concave (convex) map over a Banach space is
continuous on the interior of its domain~\cite{VanTiel}. Hence, if we can show that $E(v)$ is
finite for every $v\in X'$, then the domain $\dom(E) = X'$ and continuity follows.
The trick to show finiteness of $E(v)$ in the context of standard DFT is to
observe that $s\mapsto \calV_v(s) = (v | \rho_s) $ is \emph{relatively bounded with
$\calE_0$-bound smaller than one}, meaning that, for each $v \in X^\prime$, there 
exist $\epsilon \in (0,1)$ and $C_\epsilon \geq 0$ such that 
\begin{equation}
  |\calV_v(s)| \leq \epsilon \calE_0(s) + C_\epsilon, \quad \forall x \in S.
\end{equation}
Assuming that $\calV_v$ is relatively bounded with $\epsilon<1$  and with no
assumption on $C_\epsilon$, we obtain:
\begin{equation}
  \calE_v(s) =   \calE_0(s) + \calV_v(s) \geq \calE_0(s) - \epsilon \calE_0(s) -
  C_\epsilon = (1-\epsilon)\calE_0(s) - C_\epsilon.
\end{equation}
Taking the infimum over $s\in S$ and using the fact that $\calE_0(s)$ is by definition below
bounded, we find
\begin{equation}
  E(v) \geq (1-\epsilon)M - C_\epsilon > -\infty.
\end{equation}

\begin{proposition}
  Let $\calE_0 : S \to \RR$, $E: X^\prime \to \RR \cup \{-\infty\}$, and $\calV_v: S \to \RR$ be as described above.
  If $\calV_v$ is relatively bounded with respect to
  $\calE_0$ with bound $\epsilon<1$ for each $v \in X'$, then $E$ is continuous on $X^\prime$.
\end{proposition}
\begin{proof} $E$ is everywhere finite, so $\dom(E) = X'$. Therefore $E$ is
  continuous at any point $v\in X'$.
\end{proof}

Having determined sufficient conditions for continuity of $E$, what
about differentiability? By definition, $E$ is superdifferentiable at
$v\in X'$ if 
the superdifferential $\partial^+ E(v)$ is non-empty. 
Clearly, if $E$ is
differentiable at $v$, the superdifferential is a singleton, $\partial^+
E(v) = \{ \nabla E(v) \}$. The converse is in general \emph{almost} true \cite{VanTiel}:
\begin{theorem}
  Let $X$ be a Banach space. A convex (concave) map $f:X\to\RR\cup\{-\infty,+\infty\}$
  is G\^ateaux differentiable at $x \in X$ if and only if it is
  continuous at $x$ with a unique subgradient (supergradient).
\end{theorem}
Hence a unique sub- or supergradient does not by itself guarantee
continuity and therefore not (G\^ateaux) differentiability: \emph{A subdifferential may be a singleton even if 
the function is non-differentiable.}
This subtle point illustrates the limitations of finite-dimensional intuition and has been
the source of misunderstanding in the DFT literature. For a discussion, see Ref.\;\cite{Lammert2005} and references therein.

We can finally establish a useful statement on the differentiability
of $E : X' \to \RRminf$:
\begin{proposition}\label{thm:subdiffs4}
  Let $E : X' \to \RR$ be everywhere finite and given by the Rayleigh--Ritz
  principle in Eq.~(\ref{eq:E-def}) and assume that $F_\textnormal{CS}$ 
  in Eq.~(\ref{eq:F0-def}) is expectation valued and equal to the Lieb functional (i.e., convex and lower semi-continuous).
  Then $E$ is differentiable at $v$ if
  and only if all ground states $s \in
  \operatornamewithlimits{argmin}_{s'\in S} \calE_v(s')$ supported by $v$ have the same
  ground-state density, $s\mapsto \rho \in X$. In particular, $E$ is
  differentiable at $v$ if there exists a unique ground state $s$ for
  $v$.
\end{proposition}
\begin{proof}
  Being continuous, $E$ is differentiable at $v$ if and only if
  $\partial^+ E(v)$ is a singleton. But by Proposition~\ref{thm:subdiffs3},  
  $\partial^+ E(v)$ is a singleton if and only if all ground states of
  $v$ have the same density, using the assumption that $F_\textnormal{CS}$ is the
  Lieb functional. 
\end{proof}

Remark: The message here is that given the Rayleigh--Ritz variation
principle, it is not sufficient to know that a potential $v\in X'$ has a
unique ground-state density $\rho=\rho_s$ to prove
differentiability of $E$ at $v$. To identify the unique
supergradient $\rho\in \partial^+ E(v)$ with a ground-state density,  
$F_\textnormal{CS}$ must be identical with the Lieb functional (convex and
lower semi-continuous) and expectation valued.
If the Lieb functional happens to be
different from every possible constrained-search functional,
differentiability may fail, even if $v$ has a unique ground-state
density $\rho_s$. A counterexample is given in
Sec.\,\ref{sec:dft-application}, where we consider the Hilbert space
$X_\text{H} = L^2(\RR^3)$ rather than the usual Banach space $X_\text{L} =
L^1(\RR^3)\cap L^3(\RR^3)$. Unlike $E : X_\text{L}\to\RR$, the function
$E : X_\text{H}\to \RR$ is not differentiable. 

Starting from a Hohenberg--Kohn variation principle with an arbitrary
admissible density functional, it seems hard to prove differentiability from the
sole assumption that $v$ has a unique ground-state density
$\rho_s$. Identifying the proper constrained-search functional, based on
the Rayleigh--Ritz variation principle, seems
the only way out of the problem.

\section{Application to DFT for molecular systems}
\label{sec:dft-application}

\subsection{Ground states and ground-state densities}

For $N$-electron systems, each normalized wave function has a density
$\rho \in L^1(\RR^3)$ with $\rho \geq 0$ almost everywhere and 
$\int \! \rho(\mathbf r)\,\mathrm d \mathbf r = N$. The set of states $S$ can be taken to be either the $L^2$
normalized states $\psi \in H^1_\textnormal{S}(\RR^{3N})$ (pure states) or the set of density
matrices $\Gamma = \sum_k \lambda_k \ket{\psi_k}\bra{\psi_k}$ constructed as convex combinations
from an orthonormal set $\{\psi_k\}\subset H^1_\textnormal{S}(\RR^{3N})$ with
$\lambda_k\geq 0$ and $\sum_k\lambda_k = 1$ (mixed states). In the pure-state and mixed-state cases, respectively,
constrained search gives the functionals $F_\text{LL}$ defined in Eq.~(\ref{eq:FLL-def}) and 
$F_\text{DM}$ defined in Eq.~\eqref{eq:FDM-def}:
\begin{alignat}{3}
F_\text{LL}(\rho) &= \inf_{\Psi \mapsto \rho} \calE_0(\psi), \quad  &
\calE_0(\psi) &= \braket{\psi|\hat{T}+\hat{W}|\psi}, &\quad \calV_v(\psi) &=
\braket{\psi|\hat{V}|\psi} =(\rho_\psi \vert v), \\
F_\text{DM}(\rho) &= \inf_{\Gamma \mapsto \rho} \calE_0(\Gamma), \quad  &
\calE_0(\Gamma) &= \Tr[(\hat{T}+\hat{W})\Gamma], &\quad \calV_v(\Gamma) &= \Tr(\Gamma\hat{V}) = (\rho_\Gamma \vert v),
\end{alignat}
where $\hat{V}$ is the multiplication operator associated with $v$.
It is obvious that $F_\text{DM} \leq
F_\text{LL}$. It is less obvious, but true, that there are $\rho\in L^1$ for
which $F_\text{DM}(\rho) < F_\text{LL}(\rho)$. 

In a classic paper \cite{Kato1951}, Kato demonstrated that, for each
$v \in L^2(\RR^3) + L^\infty(\RR^3)$, the $N$-electron Hamiltonian
$\hat{H} = \hat{T} + \hat{W} + \hat{V}$ is self-adjoint on
$L^2(\RR^{3N})$ with domain $H^2_\textnormal{S}(\RR^{3N})$ and bounded below. More
generally, Simon demonstrated that, for each $v\in L^{3/2}(\RR^3) +
L^\infty(\RR^3)$ (which includes all Coulomb potentials), $\calV_v$ is relatively $\calE_0$-bounded with
an arbitrarily small bound\;\cite{Simon1971}:
\begin{theorem}\label{thm:Simon}
  Let $v \in L^{3/2}(\RR^3) + L^\infty(\RR^3)$ and let $\hat{V}$ be
  the corresponding multiplication operator. Then, for any
  $\epsilon>0$, there exists a constant $C_\epsilon \geq 0$, such that for all $\psi \in
  H^1_\textnormal{S}(\RR^{3N})$ with $\|\psi\|_2=1$,
  \begin{equation}
    |\braket{\psi|\hat{V}|\psi}| \leq \epsilon
    \braket{\psi|\hat{T}+\hat{W}|\psi} + C_\epsilon.
  \end{equation}
\end{theorem}
\begin{proof}
  See Ref.~\cite{Simon1971}.
\end{proof}
In the pure-state case, therefore, each $v \in L^{3/2}(\RR^3) +
L^\infty(\RR^3)$ gives a $\calV_v$ relatively bounded
by $\calE_0$, with bound $\epsilon$ arbitrarily small. Lieb proved
that $F_\text{LL}$ and $F_\text{DM}$ are both
expectation valued~\cite{Lieb1983}. In the same publication, a proof
(due to Simon) that $F_\text{DM} : L^1(\RR^3) \to \RRinf$ is convex and
lower semi-continuous was given.
Using the Sobolev embedding theorem, it can be shown that for each 
$\psi \in H^1_\textnormal{S}(\RR^{3N})$, the corresponding density $\rho \mapsfrom \psi$ belongs to 
$L^3(\RR^3)$ \cite{Lieb1983}. Thus, it is appropriate to consider the
Banach space 
\begin{equation}
  X_\text{L}  := L^1(\RR^3)\cap L^3(\RR^3) 
\end{equation}
with norm $\Vert \cdot \Vert = \|\cdot\|_{L^1} + \|\cdot\|_{L^3}$,
whose dual space is 
\begin{equation}
  X_\text{L}' = L^{3/2}(\RR^3) + L^{\infty}(\RR^3)
\end{equation}
with the topology induced by $\|\cdot\|$ \cite{Liu1968}. %In particular, the dual space contains all Coulomb potentials.
Since convergence in
$X_\text{L}$ implies convergence in $L^1$, 
$F_\text{DM} : X_\text{L} \to \RRinf$ is lower semi-continuous as well.
Applying Proposition~\ref{thm:subdiffs3} for $F_\text{DM}$ (which is expectation valued and lower semi-continuous convex) and
Proposition~\ref{thm:subdiffs1} for $F_\text{LL}$ (which is expectation valued but not lower semi-continuous convex), we obtain
\begin{corollary}
\label{thm:characterization} 
  Let $v\in X_{\textnormal L}^\prime  = L^{3/2}(\RR^3) + L^\infty(\RR^3)$. Then,
  \begin{alignat}{2}
 \Gamma \in \argmin_{\Gamma^\prime}\, \Bigl(\calE_0(\Gamma^\prime) +  \left(\rho_{\Gamma^\prime} \vert\, v \right) \Bigr) &\iff -v \in \partial^-
    F_\mathrm{DM}(\rho_\Gamma) &&\iff \rho_\Gamma \in \partial^+ E(v), \\
 \psi \in \argmin_{\psi^\prime}\, \Bigl(\calE_0(\psi^\prime) +  \left(\rho_{\psi^\prime} \vert\, v \right) \Bigr) &\iff -v \in \partial^-
    F_\mathrm{LL}(\rho_\psi) &&\implies \rho_\psi \in \partial^+ E(v).
  \end{alignat}
\end{corollary}
Note that $\rho \in \partial^+ E(v)$
does not imply the existence of a ground-state wave function for $v$ such
that $\psi\mapsto\rho$, only the existence of a mixed ground state $\Gamma \mapsto \rho$. 
For example, if $v \in X_\text{L}^\prime$ has a two-fold degeneracy with pure ground states $\psi_1$ and $\psi_2$,
then $\lambda\rho_1 + (1-\lambda)\rho_2 \in \partial^+ E(v)$ but there may
be no pure ground state with this density. On the other hand,
$\lambda\ket{\psi_1}\bra{\psi_1} + (1-\lambda)\ket{\psi_2}\bra{\psi_2}$
is a mixed ground state with this density.

\subsection{Topology dependence of the Lieb functional}
\label{sec:topology}

Recall that the Lieb functional $F : X \to \RRinf$ depends explicitly
on the Banach space $X$, including its topology. We now demonstrate
this using an explicit example. Motivated by the set inclusion
\begin{equation}
X_\text{L} = L^1(\RR^3)\cap L^3(\RR^3)\subset X_\text{H} := L^2(\RR^3) 
\end{equation}
and the fact that $X_\text{H}$ is Hilbert space with a simpler
structure than the non-reflexive space $X_\text{L}$, it is tempting
to consider $F_\text{DM}$ as a function on $X = X_\text{H}$.
However, the embedding $X_\text{L} \subset X_\text{H}$ is not continuous: Convergence in
$X_\text{L}$ does not imply convergence in $X_\text{H}$. Even if $F_\text{DM}$ is
lower semi-continuous with respect to $X_\text{L}$, it may fail to be
so in $X_\text{H}$. To see this, we note that, since $X_\text{H}$ is a
Hilbert space, $X_\text{H}'=X_\text{H}$, which does not admit constant potentials: If $v\in X_\text{H}$, then 
$v + c \notin X_\text{H}$ for each constant $c \neq 0$. This observation allows us
to prove the following result:
\begin{proposition}\label{thm:nonpos}
  For each $v\in X_\textnormal{H}$, we have $E(v) \leq 0$. If $v\geq 0$ almost everywhere, then $E(v) = 0$. 
\end{proposition}
\begin{proof}
  See Appendix A.
\end{proof}
The proof is based on the following idea. Write $v = v_+ - v_-$ where
$v_{\pm} \geq 0$ almost everywhere. The negative part $v_-$ can only
lower the energy, whereas $v_+\in X_\text{H}$ implies that $v$ decays
at infinity (in an average sense), thereby allowing the electrons to
lower their energy to zero by ``escaping to infinity''. Note how the
fact that nonzero $c\notin X_\text{H}$ affects this argument.

Consider now $\rho \equiv 0$, for which $F_\text{DM}(\rho) = +\infty$
since no $\Gamma \mapsto 0$. Evaluating the conjugate (Lieb) functional 
with respect to $X_\text{H}$ at $\rho=0$, we obtain
\begin{equation}
  F(0) = \sup_{v\in X_\text{H}} \left( E(v) - (v|0) \right) = \sup_{v\in X_\text{H}} E(v) = 0.
\end{equation}
Note that a different result is obtained using $X_\text{L}$, which
admits constant potentials $v(\mathbf r) \equiv c \in \mathbb R$:
\begin{equation}
  F(0) = \sup_{v\in X_\text{L}} \left( E(v) - (v|0) \right) = \sup_{v\in X_\text{L}} E(v) = +\infty = F_\text{DM}(0).
\end{equation}
As an immediate consequence, we arrive at the following result on $X_\text{H}$:
\begin{proposition}
\label{thm:subdiffsXX} 
  For any $v \in X_\textnormal{H}$ such that $v\geq 0$ almost everywhere, it holds that,
  \begin{equation}
    E(v) = F(0) + (v|0) \; \wedge \;  -v \in \partial^- F(0)
    \;\wedge\; 0 \in \partial^+ E(v), \label{eq:itholds}
  \end{equation}
where $F$ is conjugate to $E$.
  There exists no ground state $\psi \in
  H^1_\textnormal{S}(\RR^{3N})$ for any $v \geq 0$.
\end{proposition}
\begin{proof} Eq.~\eqref{eq:itholds} follows from Proposition\,\ref{thm:subdiffs-conj}.
  It remains to show that there exists no ground state if $v \geq
  0$ almost everywhere. Since $E(v) = 0$, it is sufficient to show that, for each $\psi \in
  H^1_\textnormal{S}(\RR^{3N})$, it holds that $\calE_v(\psi) > 0$. We have $\calE_v(\psi) \geq
  \braket{\psi|\hat{T}|\psi}$. But if $\psi$ has zero kinetic energy, then $\nabla\psi = 0$, from which
  it follows that $\psi = 0$ almost everywhere, contradicting
  $\|\psi\|_2 = 1$. Thus $\calE_v(\psi) \geq
  \braket{\psi|\hat{T}|\psi} > 0$.
\end{proof}

Our discussion illustrates how the choice of Banach space $X$ of densities
affects the properties of the conjugate universal functional $F : X
\to \RRinf$. Clearly, the Lieb functional $F_\text{DM}$ for the space $X_\text{L}$
is very different from the Lieb functional for $X_\text{H}$,
since, in the latter case, for any non-negative potential $v$, the unphysical density $\rho \equiv 0$
is a minimizer of the Hohenberg--Kohn variation principle. But $v$
does not even have a ground state and certainly not a ground state
with density $\rho \equiv 0$.

Connecting with the discussion following Proposition\,\ref{thm:subdiffs4},
we can see how topology influences differentiability of $E$. The map
$E: X_\text{H}\to\RR$ is pointwise identical to $E:X_\text{L}'\to\RR$:
it is by definition the ground-state energy of the system in the
external potential $v$, defined via the Rayleigh--Ritz variation
principle. However, the topology on $X_\text{H}$ is different and
$F_\text{DM}$ is no longer equal to the Lieb functional. We then cannot
conclude from Proposition\,\ref{thm:subdiffs4} that $E: X_\text{H}\to\RR$ is differentiable at $v$ 
if $v\in X_\text{H}$ has a unique ground state.

\section{Application to CDFT for molecular systems in magnetic fields}
\label{sec:cdft-application}

Consider an $N$-electron system subject to an external magnetic
vector potential $\vec{A} : \RR^3\to\RR^3$, with associated magnetic
field $\vec{B}(\vec{r}) = \nabla \times \vec{A}(\vec{r})$ (in the
distributional sense). Considering $\vec{A}$ as a variable on the same
footing as the scalar potential $v$, we arrive at CDFT, where
the paramagnetic current density $\jp \in \vec{L}^1(\RR^3)$ appears as a variable together with
$\rho$ \cite{Vignale1987,Tellgren2012,Laestadius2014}.
The mathematical foundation of CDFT is not as well developed as that
of DFT. Indeed, part of the motivation for the present work
stems from a study of CDFT~\cite{Kvaal2015}. On the other hand, Laestadius
\cite{Laestadius2014} has taken important steps---proving, for example,
Proposition~\ref{thm:FVR} below.

Following Ref.~\cite{Laestadius2014}, we assume for simplicity that
the components of $\vec{A}$ are $L^\infty(\RR^3)$ functions. The
single-electron momentum operator $-\mathrm i\nabla$ is then replaced by $-\mathrm i\nabla
+ \vec{A}$, transforming the $N$-electron kinetic-energy operator
$\hat{T}$ into the corresponding $\hat{T}_\vec{A}$. Since
$\vec{A}\in\vec{L}^\infty(\RR^3)$, the kinetic-energy expectation
value is still well defined for each $\psi \in H^1_\textnormal{S}(\RR^{3N})$. In
terms of the density $\rho \in L^1(\RR^3)$ and the paramagnetic current
density $\jp \in \vec{L}^1(\RR^3)$, we obtain
\begin{equation}
  \braket{\psi|\hat{T}_\vec{A}|\psi} = \braket{\psi|\hat{T}|\psi} +
  \frac{1}{2} \int \! |\vec{A}(\vec{r})|^2 \rho(\vec{r}) \, \mathrm d \vec{r} +
  \int \! \vec{A}(\vec{r})\cdot\jp(\vec{r}) \, \mathrm d\vec{r}.
\end{equation}
The ground-state energy is thus given by
\begin{equation}
  \begin{split}
  E(v,\vec{A}) &= \inf_{\psi, \|\psi\|_2=1}
  \braket{\psi|\hat{T}_\vec{A} + \hat{W} + \hat{V}|\psi}
  = \inf_{(\rho,\jp)} \left( F_\text{VR}(\rho,\jp) + (v + \tfrac{1}{2} A^2 \,\vert\, \rho ) + ( \vec{A} \,\vert\,\jp ) \right),
\end{split}
\end{equation}
where we have introduced the \emph{Vignale--Rasolt (VR) functional} \cite{Vignale1987}
\begin{equation}
  F_\text{VR}(\rho,\jp) := 
  \inf \left\{ \braket{\psi|\hat{T} + \hat{W}|\psi} \;
    \big|\; \psi \in H^1_\text{S}(\RR^{3N}), \; \|\psi\|_2 = 1, \; \psi \mapsto (\rho,\jp)
  \right\}.\label{eq:vr-fun}
\end{equation}
Similarly, we may define a density-matrix constrained-search functional
\begin{equation}
  F_\text{DM}(\rho,\jp) = \inf \left\{ \Tr((\hat{T} + \hat W )\Gamma) \; \big| \;
    \Gamma \mapsto (\rho,\jp) \right\}, \label{eq:FDM2-def}
\end{equation}  
where $\rho_\Gamma = \sum_k p_k \rho_{\psi_k}$ and $\jp_\Gamma =
\sum_k p_k \jp_{\psi_k}$. Whereas $F_\text{DM}$ is convex by construction,
the presence of the nonlinear $\vec{A}$-dependent term makes $E(v,\vec A)$ 
nonconcave. Following Ref.\;\cite{Tellgren2012}, it is therefore natural instead to work with 
$\tilde E(u,A) =  E(u - \tfrac{1}{2}A^2,\mathbf A)$, defined in
Proposition~\ref{thm:tildeE} below.
We note that 
\begin{equation}
(\rho,\jp) \in X_\text{L}\times \vec{L}^1(\RR^3),
\end{equation}
a Banach space with norm $\|(\rho,\jp)\| = \|\rho\| +
\|\jp\|_{\vec{L}^1}$ and dual $X_\text{L}'\times
\vec{L}^\infty(\RR^3)$, the latter containing all
$(v,\vec{A})$. 

First, we prove finiteness of the ground-state energy:
\begin{proposition}
  For every $(v,\vec{A}) \in X_\textnormal{L}\times\vec{L}^\infty$, the ground-state energy is finite,
  $E(v,\vec{A}) > -\infty$.
\end{proposition}
\begin{proof}
  For an $N$-electron system influenced by a magnetic field, the \emph{diamagnetic inequality} \cite{LiebAndLoss} gives
  \begin{equation} 
    \braket{\phi|\hat{T}|\phi} \leq
    \braket{\psi|\hat{T}_\vec{A}|\psi}, \quad \forall \psi \in H^1_\textnormal{S}(\RR^{3N}),
  \end{equation}
  where $\phi = |\psi|$ be the pointwise absolute value. 
  Next, using the identity $\rho_\psi = \rho_\phi$ and the relative boundedness
  of $v\in X_\text{L}$ with respect to $\hat{T}$, we find that
  $\hat{V}$ is relatively bounded with respect to $\hat{T}_\vec{A}$:
  \begin{equation}
    |\braket{\psi|\hat{V}|\psi}| = |\braket{\phi|\hat{V}|\phi}| \leq
    \epsilon \braket{\phi|\hat{T}|\phi} + C_\epsilon \leq \epsilon
    \braket{\psi|\hat{T}_\vec{A}|\psi} + C_\epsilon.
  \end{equation}
  Similarly, $\hat{W}$ is relatively bounded with respect to
  $\hat{T}_\vec{A}$. It follows that $\braket{\psi|\hat{H}|\psi}$ is
  below bounded, and thus that $E(v,\vec{A}) > -\infty$.
\end{proof}
Next, we note that, for each $\vec{A} \in \vec{L}^\infty(\RR^3)$, we have
$|\vec{A}|^2 \in L^\infty(\RR^3) \subset X_\text{L}'$. Therefore, for each pair 
$(v,\vec{A})\in X_\text{L}$, we have $(v \pm \tfrac{1}{2}|\vec{A}|^2,\vec{A})
\in X_\text{L} \times \vec{L}^\infty$, making the following
proposition easy to prove:
\begin{proposition}\label{thm:tildeE}
  If $F_0 : X_\text{L}\times\vec{L}^1$ is either 
  $F_\mathrm{VR}$ or $F_\mathrm{DM}$, then $\tilde{E}
  : X_\textnormal{L}' \times \vec{L}^\infty(\RR^3) \to \RR$ defined by
  \begin{equation}
    \tilde{E}(u,\vec{A}) = \inf_{(\rho,\jp)} \left( F_0(\rho,\jp) + (u\,\vert\,\rho)
    + (\vec{A}\,\vert\,\jp) \right)
  \end{equation}
  is concave, finite, and therefore continuous. The ground-state energy $E$ and $\tilde{E}$ are related by
  \begin{equation}
    E(v,\vec{A}) = \tilde{E}(v + \tfrac{1}{2}|\vec{A}|^2, \vec{A}), \quad
    \tilde{E}(u,\vec{A}) = E(u - \tfrac{1}{2}|\vec{A}|^2, \vec{A}).
  \end{equation}
\end{proposition}
\begin{proof}
  We leave the details to the reader.
\end{proof}

The map $(v,\vec{A}) \mapsto (v + \tfrac{1}{2}|\vec{A}|^2,\vec{A})$
defined on $X_\text{L}'\times\vec{L}^\infty(\RR^3)$ is smooth and
invertible, with smooth inverse $(u,\vec{A}) \mapsto (u -
\tfrac{1}{2}|\vec{A}|^2,\vec{A})$. Thus, the properties of $E$ are
reflected in properties $\tilde{E}$ and vice versa. If $F_0$ is an
admissible functional for $\tilde{E} : X_\text{L}'\times\vec{L}^\infty\to\RR$,
then
\begin{equation}
  \begin{split}
    E(v,\vec{A}) = F_0(\rho,\jp) + (v + \tfrac{1}{2}|\vec{A}|^2 | \rho
    ) + (\vec{A}|\jp) \quad &\iff \quad -(v + \tfrac{1}{2}|\vec{A}|^2,
    \vec{A}) \in \partial^- F_0(\rho,\jp) \\ &\implies\quad
    (\rho,\jp)
    \in \partial^+\tilde{E}(v+\tfrac{1}{2}|\vec{A}|^2,\vec{A})
  \end{split}
\end{equation}

Laestadius proved the following result \cite{Laestadius2014}:
\begin{proposition}\label{thm:FVR}
  $F_\mathrm{VR} : L^1(\RR^3)\times\vec{L}^1(\RR^3)\to\RRinf$ is expectation
  valued.
\end{proposition}
Proposition~\ref{thm:subdiffs1} therefore gives
\begin{equation}
  \psi_\rho \in \argmin_{\psi'} \braket{\psi|\hat{H}|\psi} \quad \iff
  \quad E(v,\vec{A}) = F_\text{VR}(\rho,\jp) + (v+\tfrac{1}{2}|\vec{A}|^2|\rho) +
  (\vec{A}|\jp).
\end{equation}
Thus, using the Vignale--Rasolt functional and the space
$X_\text{L}\times\vec{L}^1$ as density space, there are no minimizers of the pure-state Hohenberg--Kohn variation
principle for CDFT that do not correspond to ground-state wave functions.

On the other hand, it is not known whether
the density-matrix functional $F_\text{DM}$ in~(\ref{eq:vr-fun}) is
expectation valued, but it seems likely. If $F_\text{DM}$ is
\emph{not} expectation valued, there may be $(\rho,\jp)$ such that
\begin{equation}
  E(v,\vec{A}) = F_\text{DM}(\rho,\jp) + (v +
  \tfrac{1}{2}|\vec{A}|^2|\rho) + (\vec{A}|\jp) \label{eq:cdft-min}
\end{equation}
but such that there is no $\Gamma\mapsto(\rho,\jp)$, i.e.,
$(\rho,\jp)$ must be considered an unphysical minimizer of the
density-matrix Hohenberg--Kohn variation principle for CDFT.

Neither is it known whether $F_\text{DM}$ is
lower semi-continuous in the $L^1\times\vec{L}^1$ topology; the proof
for the ordinary DFT case in Ref.~\cite{Lieb1983} is not easy to
generalize to the present situation. Thus, there could be
$(\rho,\jp)\in \partial^+\tilde{E}(v+\tfrac{1}{2}|\vec{A}|^2,\vec{A})$
such that $-(v + \tfrac{1}{2}|\vec{A}|^2,\vec{A})\notin \partial^-
F_\text{DM}(\rho,\jp)$, so that $(\rho,\jp)$ does not satisfy Eq.~\eqref{eq:cdft-min}.

Finally, we may consider the Lieb functional $F :
X_\text{L}\times\vec{L}^1 \to \RRinf$, defined as
\begin{equation}
    F(\rho,\jp) := \sup_{u,\vec{A}} \left( \tilde{E}(u,\vec{A}) - (u|\rho) - (\vec{A}|\jp)\right) 
= \sup_{v,\vec{A}} \left(E(v,\vec{A}) - (v+\tfrac{1}{2}|\vec{A}|^2|\rho) - (\vec{A}|\jp) \right).
\end{equation}
This functional is convex and lower semi-continuous by construction. However, unlike in standard DFT,
it is {unknown} whether $F(\rho,\jp) = F_\text{DM}(\rho,\jp)$.

From the perspective of applying Proposition~\ref{thm:subdiffs3},
thereby establishing a result like Corollary~\ref{thm:characterization} for
CDFT, one must establish \emph{both} that $F_\text{DM}$ is expectation
valued \emph{and} lower semi-continuous. Thus, we do not know at the present
time, whether $E(v,\vec{A})$ is G\^ateaux-differentiable when the
ground-state density $(\rho,\jp)$ is unique.

\section{Conclusion}
\label{sec:conclusion}

We have analyzed the relationship between ground-state
densities obtained from the Rayleigh--Ritz variation principle (by minimizing over pure-state wave functions or mixed-state density matrices)
and from the Hohenberg--Kohn variation principles (by minimizing over densities using an admissible density functional $F_0(\rho)$).
%The starting point is an abstract
%analysis of the Hohenberg--Kohn variation principle and its connection
%with the Rayleigh--Ritz variation principle via the constrained-search
%methodology. Several concepts from convex analysis were essential---in
%particular, the sub- or superdifferential of a convex or concave
%function. 
For standard DFT for molecular systems, we established, for each potential $v
\in L^{3/2}(\RR^3) + L^{\infty}(\RR^3)$, a one-to-one correspondence
between the mixed ground-state densities of the Rayleigh--Ritz
variation principle and the mixed ground-state densities of the
Hohenberg--Kohn variation principle with the Lieb density-matrix
constrained-search functional $F_\text{DM} : L^1(\RR^3) \cap
L^3(\RR^3)\to\RRinf$. A similar
one-to-one correspondence is established between the pure ground-state
densities of the Rayleigh--Ritz variation principle and the pure
ground-state densities of the Hohenberg--Kohn variation
principle with the Levy--Lieb functional $F_\text{LL} : L^1(\RR^3)\cap
L^3(\RR^3) \to \RR$. In other words, all physical ground-state densities (pure
or mixed) are recovered with these functionals and there are no false 
densities (i.e., minimizing densities not associated with a ground-state wave function).

We also noted how the topology of the underlying Banach space $X$
impinges on the results---in particular, we noted that the Lieb functional
$F: X \to \RRinf$ depends explicitly on $X$. As an illustration, 
$F \neq F_\text{DM}$ on $X = L^2(\RR^3)$ but $F = F_\text{DM}$ on $L^1(\RR^3)\cap L^3(\RR^3)$.
Finally, CDFT was discussed and 
some open problems were pointed out---for example, it is unknown whether
$F_\text{DM}(\rho,\jp)$ is lower semi-continuous in any useful topology
such as $L^1(\RR^3)\times\vec{L}^1(\RR^3)$.

\begin{acknowledgments}

This work was supported by the Norwegian Research Council through the
CoE Centre for Theoretical and Computational Chemistry (CTCC) Grant
No.\ 179568/V30 and the Grant No.\ 171185/V30 and through the European
Research Council under the European Union Seventh Framework Program
through the Advanced Grant ABACUS, ERC Grant Agreement No.\ 267683.

\end{acknowledgments}

\appendix
\section{Proof of Proposition~\ref{thm:nonpos}}

\begin{proof}[Proof of Proposition~\ref{thm:nonpos}]
Writing $v = v_+ - v_-$, we obtain
  \begin{equation}
    E(v) = \inf_\rho \left(F_\text{DM}(\rho) + (v_+ \,\vert,\rho) - (v_- \,\vert,\rho) \right) \leq
    \inf_\rho \left( F_\text{DM}(\rho) + (v_+\,\vert\,\rho) = E(v_+) \right).
  \end{equation}
  It is therefore sufficient to show that $E(v_+) \leq 0$. In fact, we
  show that $E(v_+) = 0$.

  Let $v \geq 0$ almost everywhere.
  Let $\lambda>0$ be arbitrary and let $\Omega_{k,\lambda}$ with $k\in \NN$ be disjoint 
  cubes of side length $\lambda$ such that $\cup_k \Omega_{k,\lambda} =
  \RR^3$. Each $\Omega_{k,\lambda}$ can be obtained by translation of $\Omega_{1,\lambda}$.
  Since $v\in L^2(\RR^3)$, 
  \[ \int_{\RR^3} \!\! v(\mathbf r)^2 \,\mathrm d \mathbf r =\sum_k
  \int_{\Omega_{k,\lambda}} \!\!\!\!v (\mathbf r)^2 \,\mathrm d \mathbf r< +\infty, \]
  implying that $\int_{\Omega_{k,\lambda}} \!\! v(\mathbf r)^2\,\mathrm d \mathbf r \to 0$ as $k\to\infty$.
  Let $\psi \in \calC^\infty_c(\RR^{3N})$ be arbitrary but with
  support in $\Omega_{1,1}^{N}$ so that
  $\rho_\psi \in \calC^\infty_c(\RR^3)$ has support contained in $\Omega_{1,1}$. By translating
  $\psi$ properly (denoting the result by $\psi_k$), the support of $\psi_k$ is inside
  $\Omega_{k,1}^N$ and
  \begin{equation}
    F_\text{DM}(\rho_{\psi_k}) = F_\text{DM}(\rho_\psi) \leq \left\langle \psi \vert T+W \vert \psi \right\rangle
    \equiv \braket{T} + \braket{W},
  \end{equation}
  independent of $k$.
We obtain
  \begin{equation}
    E(v) \leq \lim_k \left(\braket{T} + \braket{W} + \int_{\Omega_{k,1}} \!\!\!\!v(\mathbf r) \,\rho_k(\mathbf r) \,\mathrm d \mathbf r \right) = \braket{T} + \braket{W},
  \end{equation}
  where we have used the fact that
  \begin{equation}
    \int_{\Omega_{k,1}} \!\!\!\!v(\mathbf r)\,\rho_k(\mathbf r)\,\mathrm d\mathbf r 
\leq \left(\int_{\Omega_{k,1}}\!\!\!\!v(\mathbf r)^2\, \mathrm d\mathbf r \right)^{1/2} \, \|\rho\|_2 \to 0
  \end{equation}
  as $k\to\infty$.
  
  We now increase the size of the boxes $\Omega_{k,\lambda}$ by varying $\lambda>0$.
  By dilating $\psi$ in the manner 
  \begin{equation}
    \psi(\vec{r}_1,\cdots) \mapsto \lambda^{3N/2}
    \psi(\lambda\vec{r}_1,\cdots),
  \end{equation}
  the support is still inside $\Omega_{1,\lambda}^\infty$ and the density
  is scaled as $\rho_\psi(\vec{r})\to
  \lambda^{-3}\rho_\psi(\lambda^{-1}\vec{r})$, conserving the number of particles.
We obtain the scaling
  \begin{equation}
    \braket{T} + \braket{W} \to \lambda^{-2}\braket{T} +
    \lambda^{-1}\braket{W}.
  \end{equation}
  By repeating the above argument for $\lambda=1$ and letting
  $\lambda\to\infty$, we obtain $E(v)\leq 0$. On the other hand, $E(v)
  \geq 0$ since the Hamiltonian $H(v)$ is positive with $v \geq 0$, yielding $E(v) = 0$.
\end{proof}

\bibliography{rr-hk-manuscript.bib}

\end{document}